%% file: root.tex
\documentclass
{llncs}  

\newif\ifieee
\ieeefalse		

\newif\ifllncs
\llncstrue		

\newif\ifamsart
\amsartfalse             

\newif\ifproofs
\proofstrue             

\newif\ifpubrel
\pubrelfalse            

\newif\ifarchive
\archivetrue            

\newif\ifieeetrans
\ieeetransfalse 

\newif\ifart
\artfalse               

\newif\ifpriorities
\prioritiestrue 

\ifllncs\usepackage{smallsubsub}\fi
\input{macros}

\usepackage{enumerate}
\usepackage[all]{xy}\SelectTips{cm}{}
\usepackage{fancyhdr}

\ifieeetrans
\newcommand{\para}[1]{\emph{#1}}
\else
\newcommand{\para}[1]{\paragraph{#1}}
\fi

\newtheorem{goal}{Security Goal}

\ifpriorities

\else

\fi

\fancypagestyle{title}{%
  \fancyhf{}%
  \fancyhead[L]{Approved for Public Release; Distribution Unlimited. Case Number XX-XXXX.}%
  \fancyfoot[C]{\copyright 2015 The MITRE Corporation. ALL RIGHTS RESERVED.}%
}

\title{Formal Support for Standardizing\\Protocols with State}

\author{Joshua D. Guttman \and Moses D. Liskov \and \\ John
  D. Ramsdell \and Paul D. Rowe}
\institute{The MITRE Corporation
}

\begin{document}
\maketitle
\ifpubrel\thispagestyle{title}\fi
\input{abstract}
\input{intro}
\input{envelope}
\input{state_axioms}
\input{state_based_analysis}
\input{axiom2}
\input{related}
\input{goals}

\nocite{RamsdellEtAl09}

\bibliography{../../inputs/secureprotocols}
\bibliographystyle{plain}



\end{document}

%% file: macros.tex
\usepackage{latexsym}
\usepackage{epic}
\usepackage{epsfig}
\usepackage{amsmath}
\usepackage{amssymb}
\usepackage{amscd}
\usepackage{url}
\usepackage{color}

\usepackage{graphics}

\input{xy}
\xyoption{all}

\newcommand{\cpsa}{\textsc{cpsa}}
\newcommand{\kind}[1]{\ensuremath{\mathsf{#1}}}

\newcommand{\ran}{\kind{ran}}

\newcommand{\skel}{\ensuremath{\mathbb{A}}}
\newcommand{\skelB}{\ensuremath{\mathbb{B}}}
\newcommand{\skelC}{\ensuremath{\mathbb{C}}}

\ifamsart{}\else{

}\fi

\newcommand{\qdot}{\,\mathbf{.}\;}

\newcommand{\limp}{\Longrightarrow}

\newcommand{\cf}[1]{\ensuremath{\mathsf{cf}( #1 )}}

\newcommand{\bnd}{\mathcal{B}}

\newcommand{\seq}[1]{\langle #1 \rangle}

\newcommand{\length}[1]{\kind{length}( #1 )} 

\newcommand{\enc}[2]{\{\!|#1|\!\}_{#2}}

\newcommand{\fn}[1]{\ensuremath{\operatorname{\mathit{#1}}}}

\newcommand{\xstinit}{\mathsf{init}}
\newcommand{\xtran}{{\mathsf{tran}}}
\newcommand{\xobsv}{\mathsf{obsv}}

\newcommand{\stinit}{\xstinit\,}
\newcommand{\tran}{\xtran\,}
\newcommand{\obsv}{\xobsv\,}

\newcommand{\compute}{\ensuremath\mathcal{C}}

\newcommand{\cn}[1]{\ensuremath{\mathop{\smash{\sf#1}}\!\mathop{\vphantom{#1}}\nolimits}}

\newcommand{\msg}{\cn{msg}}
\newcommand{\nodes}{\ensuremath{\cn{nodes}}}
\newcommand{\syncnodes}{\ensuremath{\cn{sync}}}
\newcommand{\states}{\cn{St}}
\newcommand{\initstates}{\mathsf{Init}\!\states}
\newcommand{\trans}{\rhd}

\newcommand{\fgt}{\ensuremath{\cn{fgt}}}

\newcommand{\Pos}{\mathsf{Pos}}

\ifllncs
\newtheorem{prop}{Proposition}{\bfseries}{\itshape}
\newtheorem{principle}{Principle}{\bfseries}{\itshape}
\newtheorem{mylemma}[prop]{Lemma}{\bfseries}{\itshape}
\newtheorem{thm}{Theorem}{\bfseries}{\itshape}
\newtheorem{eg}{Example}{\bfseries}{\upshape}
\newtheorem{cor}[prop]{Corollary}{\bfseries}{\itshape}
\spnewtheorem*{proofsketch}{Proof Sketch}{\itshape}{\rm}
\spnewtheorem*{thmredux}{Theorem}{\bfseries}{\itshape}
\spnewtheorem*{propredux}{Proposition}{\bfseries}{\itshape}

\else
\newtheorem{prop}{Proposition}[section]

\newtheorem{definition}[prop]{Definition}

\newtheorem{lemma}[prop]{Lemma}

\ifamsart{}\else
\def\squareforqed{\hbox{\rlap{$\sqcap$}$\sqcup$}}
\def\qed{\ifmmode\squareforqed\else{\unskip\nobreak\hfil
\penalty50\hskip1em\null\nobreak\hfil\squareforqed
\parfillskip=0pt\finalhyphendemerits=0\endgraf}\fi}

\newenvironment{proof}{\rm {\noindent \bfseries Proof:}}{\hspace*{\fill}\qed\par\smallskip}
\fi 
\fi

\newcommand{\srt}[1]{\ensuremath{\mathsf{#1}}}

\renewcommand{\seq}[1]{\ensuremath{\langle#1\rangle}}



%% file: abstract.tex
\begin{abstract}
Many cryptographic protocols are designed to achieve their goals using
only messages passed over an open network. Numerous tools, based on
well-understood foundations, exist for the design and analysis of
protocols that rely purely on message passing. However, these tools
encounter difficulties when faced with protocols that rely on
non-local, mutable state to coordinate several local sessions. 

We adapt one of these tools, {\cpsa}, to provide automated support for
reasoning about state. We use Ryan's Envelope Protocol as an example
to demonstrate how the message-passing reasoning can be integrated
with state reasoning to yield interesting and powerful results.

\medskip
{\bf Keywords:} protocol analysis tools, stateful protocols, TPM, PKCS \#11.
\end{abstract}


%% file: intro.tex
\section{Introduction}
\label{sec:intro}

Many protocols involve only message transmission and reception,
controlled by rules that are purely local to a session of the
protocol.  Typical protocols for authentication and key establishment
are of this kind; each participant maintains only the state required
to remember what messages must still be transmitted, and what values
are expected in messages to be received from the peer.  

Other protocols interact with long-term state, meaning state that
persists across different sessions and may control behavior in other
sessions.  A bank account is a kind of long-term state, and it helps
to control the outcome of protocol sessions in the ATM network.
Specifically, the session fails when we try to withdraw money from an
empty account.  Of course, one session has an effect on others through
the state:  When we withdraw money today, there will be less remaining
to withdraw tomorrow.  

Hardware devices frequently participate in protocols, and maintain
state that helps control those protocols.  For example, PKCS\#11
devices store and use keys, and are constrained by key attributes that
control e.g.~which keys may be used to wrap and export other keys.
Trusted Platform Modules (TPMs) maintain Platform Configuration
Registers (PCRs) some of which are modified only by certain special
instructions.  Thus, digitally signing the values in these registers
attests to the history of the platform.  Some protocols involve
multiple state histories; for instance, an online bank transfer
manipulates the state of the destination account as well as the state
of the source account.


State-based protocols are more challenging to analyze than protocols
in which all state is session-local.  Among the executions that are
possible given the message flow patterns, one must identify those for
which a compatible sequence of states exists.  Thus, to justify
standardizing protocols involving PKCS\#11 devices or TPMs, one must
do a deeper analysis than for stateless protocols.  Indeed, since
these devices are themselves standardized, it is natural to want to
define and justify protocols that depend only on their required
properties, rather than any implementation specific peculiarities.  

The goal of this paper is to explain formal ideas that can automate
this analysis, and to describe a support tool that assists with it.

\paragraph{Contributions of this paper.}  We make four main
contributions:
\begin{itemize}
  \item We identify two central axioms of state that formalize the
  semantics of state-respecting behaviors 
  (Def.~\ref{def:bundle:state:respecting}).  Each time a state is
  produced,
  \begin{enumerate}
    \item it can be consumed by at most one subsequent transition.
    \item it cannot be observed after a subsequent transition consumes
    it.
  \end{enumerate}

  The first axiom is the essence of how the state-respecting analysis
  differs from standard message-based analysis.  By contrast, once a
  message has been transmitted, it can be delivered (or otherwise
  consumed) repeatedly in the future.
  
  The second axiom, like the reader/writer principle in concurrency,
  allows observations to occur without any intrinsic order among them,
  so long as they all occur while that state is still available.  It
  preserves the advantages of a partial order model, as enriched with
  state.
  \item An alternative model of execution maintains state in a family
  of traditional state machines, whose transitions are triggered by
  synchronization events in a state-respecting
  manner\ifarchive{}\fi.
  The justification for our two axioms is that they match this
  alternative, explicit-state-machine model exactly.  \ifarchive{We 
    prove this in
    Lemmas~\ref{lemma:execution:to:sr:bundle}--\ref{lemma:sr:bundle:to:execution}.}
  \fi
  \item We incorporated these two axioms into the tool
  {\cpsa}~\cite{cpsa09}, obtaining a tool that can perform
  state-respecting enrich-by-need protocol analysis.
  \item We applied the resulting version of {\cpsa} to an interesting
  TPM-based protocol, the Envelope
  Protocol~\cite{arapinis2011statverif}, verifying that it meets its
  security goal.  We have also analyzed some incorrect variants,
  obtaining attacks.
\end{itemize}

\paragraph{Roadmap.}  After giving some background, we describe the
Envelope Protocol and the TPM behaviors it relies on
(Section~\ref{sec:envelope}).  We introduce our protocol model
(Section~\ref{sec:state_axioms}) in both its plain form, and the form
enriched by the axioms in Contribution~1.
Section~\ref{sec:state_based_analysis} describes the {\cpsa} analysis
in the original model where state propagation is not distinguished
from message-passing, and in the enriched model.  We turn to related
work in Section~\ref{sec:related}.  Section~\ref{sec:goals} addresses
a logical interpretation of enrich-by-need analysis and observes that
this framework may be used, unmodified, for stateful protocols as we
model them.  We end with a brief comment on conclusions and future
work.

\paragraph{Background:  Strand spaces.}  We work within the strand
space framework.  A \emph{strand} is a (usually short) finite
sequence of events, where the events are
\begin{description}
  \item[message transmission] nodes; 
  \item[message reception] nodes; and 
  \item[state synchronization] nodes.  
\end{description}
Each message transmission and reception node is associated with a
message that is sent or received.  State synchronization nodes were
introduced into strand spaces recently~\cite{Guttman12}.  Including
them des not alter key definition such as bundles
(Def.~\ref{def:bundle}), and they allow us to flag events that, though
the protocol principals perform them, are not message events.  State
synchronization nodes will be related to states via two different
models in Section~\ref{sec:state_axioms}.

The behavior of a principal in a single, local run of one role of a
protocol forms a strand.  We call these \emph{regular strands}.  We
also represent basic actions of an adversary as strands, which we call
\emph{adversary strands}.  Adversary strands never need state
synchronization nodes, since our model of the adversary allows it to
use the network as a form of storage that never forgets old messages.  

A protocol $\Pi$ is represented by a finite set of strands, called the
\emph{roles} of the protocol, together with some auxiliary information
about freshness and non-compromise assumptions about the roles.  We
write $\rho\in\Pi$ to mean that $\rho$ is one of the roles of the
protocol $\Pi$.  The \emph{regular strands of} $\Pi$ are then all
strands that result from any roles $\rho\in\Pi$ by applying a
substitution that plugs in values in place of the parameters occurring
in $\rho$.

For more information on strand spaces, see
e.g.~\cite{Guttman10,ThayerHerzogGuttman99}.  For the version
containing state synchronization events as well as transmissions and
receptions, see~\cite{Guttman12,RamsdellDGR14}.

\paragraph{Background:  Enrich-by-need analysis.}  In our form of
protocol analysis, the input is a fragment of protocol behavior.  


The output gives zero or more executions that contain this fragment.
We call this approach ``enrich-by-need'' analysis (borrowed from
our~\cite{Guttman14}), because it is a search process that gradually
adds information as needed to explain the events that are already
under consideration.

An analysis begins with an execution fragment $\skel$, which may, for
instance, reflect the assumption that one participant has engaged in a
completed local session (a strand); that certain nonces were freshly
chosen; and that certain keys were uncompromised.  The result of the
analysis is a set $S$ of executions enriching the starting fragment
$\skel$.  An algorithm implementing this approach is sound if, for
every possible execution $\skelC$ that enriches $\skel$, there is a
member $\skelB\in S$ such that $\skelC$ enriches $\skelB$.

We do not require $S$ to contain all possible executions because there
are infinitely many of them if any.  For instance, executions may
always be extended by including additional sessions by other protocol
participants.  Thus, we want the set $S$ to contain representatives
that cover all of the \emph{essentially different} possibilities.  We
call these representatives $S$ the \emph{shapes} for $\skel$.

In practice, the set $S$ of shapes for $\skel$ is frequently finite
and small.  

When we start with a fragment $\skel$ and find that it has the empty
set $S=\emptyset$ of shapes, that means that no execution contains all
of the structure in $\skel$.  To use this technique to show
confidentiality assertions, we include a disclosure event in $\skel$.
If $\skel$ extends to no possible executions at all, we can conclude
that this secret cannot be revealed.  If $S$ is non-empty, the shapes
are attacks that show how the confidentiality claim could fail.

The set $S$ of shapes, when finite, also allows us to ascertain
whether authentication properties are satisfied.  If each shape
$\skelB\in S$ satisfies an authentication property, then every
possible execution $\skelC$ enriching $\skel$ must satisfy the
property too:  They all contain at least the behavior exhibited in
some shape, which already contained the events that the authentication
property required.

This style of analysis is particularly useful in a partially ordered
execution model, such as the one provided by strand spaces.  In
partially ordered models, when events $e_1,e_2$ are causally
unrelated, neither precedes the other.  In linearly ordered execution
models, both interleavings $e_1\prec e_2$ and $e_2\prec e_1$ are
possible, and must be considered.  When there are many such pairs,
this leads to exponentially many interleavings.  None of the
differences between them are significant.


%% file: envelope.tex
\section{The Envelope Protocol}
\label{sec:envelope}

We use Mark Ryan's Envelope Protocol~\cite{ArapinisEtAl2011} as a
concrete example throughout the paper. The protocol leverages
cryptographic mechanisms supported by a TPM to allow one party to
package a secret such that another party can either reveal the secret
or prove the secret never was and never will be revealed, but not
both.

It is a particularly useful example to consider because it is
carefully designed to use state in an essential way. In particular, it
creates the opportunity to take either of two branches in a state
sequence, but not both. In taking one branch, one loses the option to
take the other. In this sense, it utilizes the non-monotonic nature of
state that distinguishes it from the monotonic nature of
messages. Additionally, although the Envelope Protocol is not
standardized, it demonstrates advanced and useful ways to use the TPM.
Standardization of such protocols is under the purview of the Trusted
Computing Group (TCG). It will be very useful to understand the
fundamental nature of state and to provide methods and tools to
support the future standardization of protocols involving devices such
as the TPM.

\para{Protocol motivation.}  The plight of a teenager motivates
the protocol.  The teenager is going out for the night, and her
parents want to know her destination in case of emergency.  Chafing at
the loss of privacy, she agrees to the following protocol.  Before
leaving for the night, she writes her destination on a piece of paper
and seals the note in an envelope.  Upon her return, the parents can
prove the secret was never revealed by returning the envelope
unopened. Alternatively, they can open the envelope to learn her
destination.

The parents would like to learn their daughter's destination while
still pretending that they have respected her privacy. The parents are
thus the adversary.  The goal of the protocol is to prevent this
deception.

\para{Necessity of long-term state.}  The long-term state is the
envelope.  Once the envelope is torn, the adversary no longer has
access to a state in which the envelope is intact.  A protocol based
only on message passing is insufficient, because the ability of the
adversary monotonically increases. Initially,
the adversary has the ability to either return the envelope or tear
it. In a purely message-based protocol the adversary will never lose
these abilities.

\para{Cryptographic version.}  The cryptographic version of this
protocol uses a TPM to achieve the security goal.  Here we restrict
our attention to a subset of the TPM's functionality. In particular we
model the TPM as having a state consisting of a single PCR and only
responding to five commands.

A \texttt{boot} command (re)sets the PCR to a known value.  The
\texttt{extend} command takes a piece of data, $d$, and replaces the
current value $s$ of the PCR state with the hash of $d$ and
$s$, denoted $\#(d,s)$.  In fact, the form of
\texttt{extend} that we model, which is an \texttt{extend} within an
encrypted session, also protects against replay.  These are the only
commands that alter the value in a PCR.

The TPM provides other services that do not alter the PCR.  The
\texttt{quote} command reports the value contained in the PCR and is
signed in a way as to ensure its authenticity.  The \texttt{create
  key} command causes the TPM to create an asymmetric key pair where
the private part remains shielded within the TPM.  However, it can
only be used for decryption when the PCR has a specific value.  The
\texttt{decrypt} command causes the TPM to decrypt a message using
this shielded private key, but only if the value in the PCR matches
the constraint of the decryption key.

In what follows, Alice plays the role of the teenaged daughter
packaging the secret. Alice calls the \texttt{extend} command with a
fresh nonce $n$ in an encrypted session.  She uses the \texttt{create
  key} command constraining a new key $k'$ to be used only when a
specific value is present in the PCR.  In particular, the constraining
value $cv$ she chooses is the following:
$$ cv = \#(\cn{\mathtt{\cn{obt}}},\#(n,s)) $$ where $\cn{obt}$ is a
string constant and $s$ represents an arbitrary PCR value prior the
extend command.  She then encrypts her secret $v$ with $k'$, denoted
$\enc{v}{k'}$.

Using typical message passing notation, Alice's part of the protocol
might be represented as follows (where we temporarily ignore the
replay protection for the \texttt{extend} command):\\

\noindent
$$
\begin{array}{c@{{}\to{}}c@{~:~}l}
\mathrm{A} & \mathrm{TPM} & \enc{\mathtt{\cn{ext}},n}{k}\\
\mathrm{A} & \mathrm{TPM} & \mathtt{\cn{create}},
\#(\mathtt{\cn{obt}},\#(n,s)) \\
\mathrm{TPM} & \mathrm{A} & k'\\
\mathrm{A} & \mathrm{Parent} & \enc{v}{k'}\\
\end{array}
$$
The parent acts as the adversary in this protocol. We assume he can
perform all the normal Dolev-Yao operations such as encrypting and
decrypting messages when he has the relevant key, and interacting with
honest protocol participants. Most importantly, the parent can use the
TPM commands available in any order with any inputs he likes. Thus he
can extend the PCR with the string \texttt{obtain} and use the key to
decrypt the secret.  Alternatively, he can refuse to learn the secret
and extend the PCR with the string $\cn{ref}$ and then generate a TPM
quote as evidence the secret will never be exposed. The goal of the
Envelope Protocol is to ensure that once Alice has prepared the TPM
and encrypted her secret, the parent should not be able to both
decrypt the secret and also generate a refusal quote,
$\enc{\cn{\mathtt{\cn{quote}}},
  \#(\cn{\mathtt{\cn{ref}}},\#(n,s)), \enc{v}{k'}}{\fn{aik}}$.

A crucial fact about the PCR state in this protocol is the
collision-free nature of hashing, ensuring that for every $x$
\begin{align} \label{eq:injective} \#(\mathtt{\cn{obt}},\#(n,s))
  \quad \neq \quad \#(\mathtt{\cn{ref}}, x )
\end{align}

\para{Formal protocol model.}

We formalize the TPM-based version of the Envelope Protocol using
strand spaces~\cite{Guttman10}. Messages and states are represented as
elements of a crypto term algebra, which is an order-sorted quotient term algebra.
Sort~$\top$ is the top sort of messages.
Messages of sort~$\srt{A}$ (asymmetric keys), sort~$\srt{S}$
(symmetric keys), and sort~$\srt{D}$ (data) are
called \emph{atoms}.  Messages are atoms, tag constants, or
constructed using encryption $\enc{\cdot}{(\cdot)}$, hashing
$\#(\cdot)$, and pairing $(\cdot,\cdot)$, where the comma operation is
right associative and parentheses are omitted when the context
permits.

We represent each TPM command with a separate role that receives a
request, consults and/or changes the state and optionally provides a
response.  As shown in Fig.~\ref{fig:TPM roles}, we use
$m\!\!\to\!\!\bullet$ and $\bullet\!\!\to\!\!m$ to represent the
reception and transmission of message $m$ respectively. Similarly, we
use $s\!\leadsto\!\!\circ$ and $\circ\!\leadsto\!s$ to represent the
actions of reading and writing the value $s$ to the state.  We write
$m\Rightarrow n$ to indicate that $m$ precedes $n$ immediately on the
same strand.

\begin{figure}
  \begin{trivlist}\item
    \hfil\xymatrix@R=3ex@C=1.1em{
      &\texttt{[re-]boot}&\\
      \ar@{->}[r]^{\cn{boot}}&\bullet\ar@{=>}[d]&\\
          [\ar@{~>}[r]^{s}&]\circ\ar@{~>}[r]^{\cn{s}_0}&
    }\hfil
    \xymatrix@R=3ex@C=1.7em{
      &\texttt{create key}&\\
      \ar@{->}[r]^{\cn{create}, s}&\bullet\ar@{=>}[d]&\\
      &\bullet\ar@{->}[rr]^{\enc{\cn{created},k',s}{\fn{aik}}}&&
    }\hfil
    \xymatrix@R=3ex@C=2em{
      &\texttt{quote}&\\
      \ar@{->}[r]^{\cn{quote},n}&\bullet\ar@{=>}[d]&\\
      \ar@{~>}[r]^{s}&\circ\ar@{=>}[d]&\\
      &\bullet\ar@{->}[rr]^{\enc{\cn{quote},s,n}{\fn{aik}}}&&
    }\hfil\vspace{2.5ex}\item
    \hfil\xymatrix@R=3ex@C=3.2em{
      &&\texttt{extend}&\\
      \ar@{->}[rr]^{\cn{sess},\fn{tpmk},\enc{\fn{esk}}{\fn{tpmk}}}
      &&\bullet\ar@{=>}[d]&\\
      &&\bullet\ar@{=>}[d]\ar@{->}[r]^{\cn{sess},\fn{sid}}&\\
      \ar@{->}[rr]^{\enc{\cn{ext},n,\fn{sid}}{\fn{esk}}}&&\bullet\ar@{=>}[d]&\\
      \ar@{~>}[rr]^{s}&&\circ\ar@{~>}[r]^{\#(n,s)}&
    }\hfil
    \xymatrix@R=3ex{
      &&\texttt{decrypt}&\\
      \ar@{->}[rr]^{\cn{dec},\enc{m}{k'}}&&\bullet\ar@{=>}[d]&\\
      \ar@{->}[rr]^{\enc{\cn{created},k',s}{\fn{aik}}}&&\bullet\ar@{=>}[d]&\\
      \ar@{~>}[rr]^{s}&&\circ\ar@{=>}[d]&\\
      &&\bullet\ar@{->}[r]^{m}&
    }\hfil
  \end{trivlist}

  \caption{TPM roles}\label{fig:TPM roles}
\end{figure}
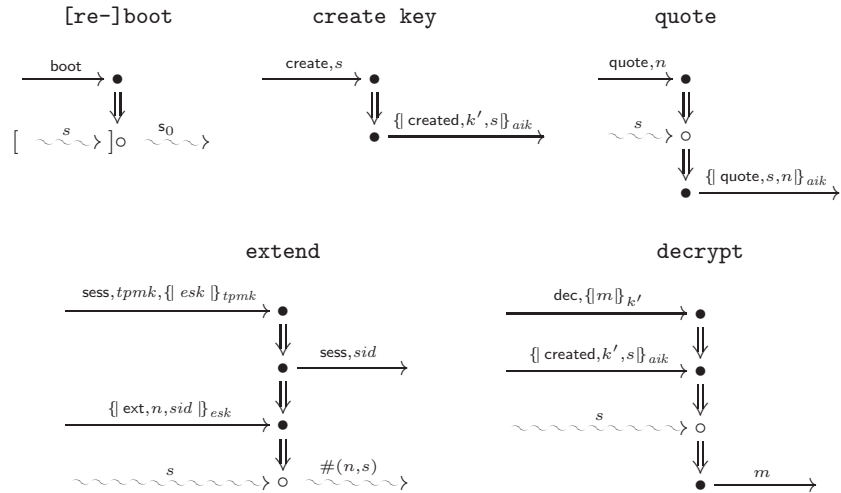

As noted above, the \texttt{boot} role and the \texttt{extend} role
are the only two roles that alter the state. This is depicted with the
single event $\leadsto\!\!\circ\!\leadsto$ that atomically reads and
then alters the state. The \texttt{boot} role receives the command and
resets any current state $s$ to the known value $\cn{s}_0$.  An
alternate version of \texttt{boot} is needed to ensure that our
sequences of state are well-founded.  This version has a single state
write event $\circ\!\leadsto\cn{s}_0$.

The \texttt{extend} role first creates an encrypted channel by
receiving an encrypted session key $\fn{esk}$ which is itself
encrypted by some other secured TPM asymmetric key $tpmk$. The TPM
replies with a random session id $sid$ to protect against replay. It
then receives the encrypted command to extend the value $n$ into the
PCR and updates the arbitrary state $s$ to become $\#(n,s)$.

The \texttt{create key} role does not interact directly with the
state. It receives the command with the argument $s$ specifying a
state. It then replies with a signed certificate for a freshly created
public key $k'$ that binds it to the state value $s$. The certificate
asserts that the corresponding private key $k'^{-1}$ will only be used
in the TPM and only when the current value of the state is $s$. This
constraint is leveraged in the \texttt{decrypt} role which receives a
message $m$ encrypted by $k'$ and a certificate for $k'$ that binds it
to a state $s$. The TPM then consults the state (without changing it)
to ensure it is in the correct state before performing the decryption
and returning the message $m$.

Finally, the \texttt{quote} role receives the command together with a
nonce $n$. It consults the state and reports the result $s$ in a
signed structure that binds the state to the nonce to protect against
replay. 

Since the \texttt{quote} role puts the state $s$ into a message, and
the \texttt{extend} role puts a message into the state, in our
formalization states are the same kind of entity as messages.  



We similarly formalize Alice's actions. Her access to the TPM state is
entirely mediated via the message-based interface to the TPM, so her
role has no state events. It is displayed in Fig.~\ref{fig:Alice role}

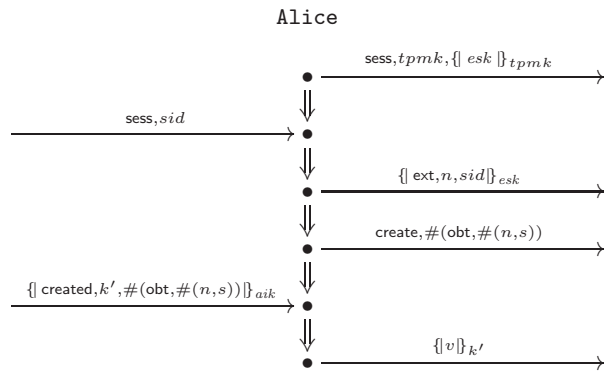
\begin{figure}
  \begin{trivlist}\item
    \hfil\xymatrix@R=3ex@C=10.5em{
      &\texttt{Alice}&\\
      &\bullet\ar@{=>}[d]\ar@{->}[r]^{\cn{sess},tpmk,\enc{\fn{esk}}{tpmk}}&\\
      \ar@{->}[r]^{\cn{sess},sid}&\bullet\ar@{=>}[d]&\\
      &\bullet\ar@{=>}[d]\ar@{->}[r]^{\enc{\cn{ext},n,sid}{\fn{esk}}}&\\
      &\bullet\ar@{=>}[d]\ar@{->}[r]^{\cn{create},
        \#(\cn{obt},\#(n,s))}&\\
      \ar@{->}[r]^{\enc{\cn{created},k',\#(\cn{obt},\#(n,s))}{\fn{aik}}~}&
      \bullet\ar@{=>}[d]&\\
      &\bullet\ar@{->}[r]^{\enc{v}{k'}}&
    }\hfil
  \end{trivlist}
  \caption{Alice's role}\label{fig:Alice role}
\end{figure}

Alice begins by establishing an encrypted session with the TPM in
order to extend a fresh value $n$ into the PCR. She then has the TPM
create a fresh key that can only be used when the PCR contains the
value $\#(\cn{obt},\#(n,s))$, where $s$ is whatever value was in the
PCR immediately before Alice performed her extend command. Upon
receiving the certificate for the freshly chosen key, she uses it to
encrypt her secret $v$ that gives her destination for the night.

The parents may then either choose to further extend the PCR with the
value $\cn{obt}$ in order to enable the decryption of Alice's secret,
or they can choose to extend the PCR with the value $\cn{ref}$ and get
a quote of that new value to prove to Alice that they did not take the
other option.  The adversary roles displayed in
Fig.~\ref{fig:adversary roles} constrain what the parents can do. 

\begin{figure}
  \begin{center}
    \begin{tabular}{ccccc}
      \xymatrix@R=3ex@C=1.2em{
        \texttt{create}\\
        \bullet\ar[r]^{a}&
      }&
      \xymatrix@R=3ex@C=1.6em{
        &\texttt{pair}\\
        \ar[r]^{x}&\bullet\ar@{=>}[d]\\
        \ar[r]^{y}&\bullet\ar@{=>}[d]\\
        &\bullet\ar^{(x,y)}[r]&
      }&
      \xymatrix@R=3ex@C=1.6em{
        &\texttt{sep}\\
        \ar[r]^{(x,y)}&\bullet\ar@{=>}[d]\\
        &\bullet\ar^{x}[r]\ar@{=>}[d]&\\
        &\bullet\ar^{y}[r]&
      }&
      \xymatrix@R=3ex@C=1.6em{
        &\texttt{enc}\\
        \ar[r]^{x}&\bullet\ar@{=>}[d]\\
        \ar[r]^{k}&\bullet\ar@{=>}[d]\\
        &\bullet\ar^{\enc{x}{k}}[r]&
      }&
      \xymatrix@R=3ex@C=1.6em{
        &\texttt{dec}\\
        \ar[r]^{\enc{x}{k}}&\bullet\ar@{=>}[d]\\
        \ar[r]^{k^{-1}}&\bullet\ar@{=>}[d]\\
        &\bullet\ar^{x}[r]&
      }
    \end{tabular}
  \end{center}
  \caption{Adversary roles, where $a$ in the \texttt{create} role must
  be an atomic message.}\label{fig:adversary roles}
\end{figure}
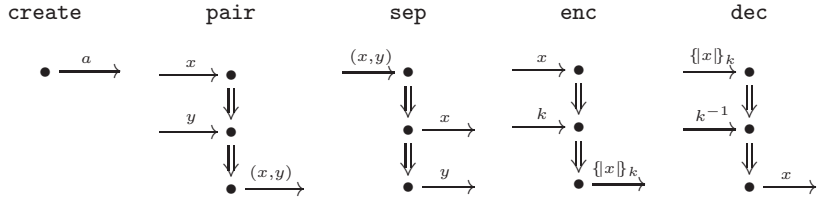

It is important to note that,
like Alice's role, the adversary roles do not contain any state
events. Thus the adversary can only interact with the state via the
interface provided by the TPM commands.

We aim to validate a particular security goal of the Envelope Protocol
using the enrich-by-need method. The parent should not be able to both
learn the secret value $v$ and generate a refusal token.

\begin{goal}\label{goal:alice}
  Consider the following events:
  \begin{itemize}
    \item An instance of the Alice role runs to completion, with
    secret $v$ and nonce $n$ both freshly chosen;
    \item $v$ is observed unencrypted;
    \item the refusal certificate
      $\enc{\cn{\mathtt{\cn{quote}}},
      \#(\cn{\mathtt{\cn{ref}}},\#(n,s)),
      \enc{v}{k'}}{\fn{aik}}$
      is observed unencrypted.
  \end{itemize}
  These events, which we call jointly $\skel_0$, are not all present
  in any execution.
\end{goal}


%% file: state_axioms.tex
\section{State-respecting bundles}
\label{sec:state_axioms}

In this section, we introduce a model of protocol behavior in the
presence of global state; it is new in this paper.  It enriches the
notion of a bundle, which is the longstanding strand space
formalization of global
behaviors~\cite{ThayerHerzogGuttman99,Guttman10}.  

We organize this section as a sequence of refinements, starting from
the traditional strand space bundle notion (Def.~\ref{def:bundle}).
We then give a direct generalization, \emph{state-enriched bundles}
(Def.~\ref{def:bundle:enriched}) to associate states with
synchronization events, and to track their propagation.  We then
\ifarchive{introduce (Def.~\ref{def:b:c:phi})} 
\fi
the notion of an \emph{execution}, which explicitly includes both a
bundle (as a global record of events and their causal ordering) and a
family of state histories, and note that state-enriched bundles are
not restrictive enough to match this notion of execution.  This
motivates the two axioms of state, leading to our final model of
stateful protocol executions, \emph{state-respecting bundles}
(Def.~\ref{def:bundle:state:respecting}), which matches the notion of
executions.  
\ifarchive{We then prove the match between the two definitions in
  Lemmas~\ref{lemma:execution:to:sr:bundle}--\ref{lemma:sr:bundle:to:execution}.}
\fi
\begin{definition}[Bundle] 
  \label{def:bundle}
  Suppose that $\Sigma$ is a finite set of strands.  Let $\Rightarrow$
  be the strand succession relation on $\nodes(\Sigma)$.  Let
  $\rightarrow\subseteq\nodes(\Sigma)\times\nodes(\Sigma)$ be any
  relation on nodes of $\Sigma$ such that $n_1\rightarrow n_2$ implies
  that $n_1$ is a transmission event, $n_2$ is a reception event, and
  $\msg(n_1)=\msg(n_2)$.

  $\bnd=(\mathcal{N},\rightarrow)$ is a \emph{bundle} over $\Sigma$
  iff $\mathcal{N}\subseteq\nodes(\Sigma)$, and
  \begin{enumerate}
    \item If $n_2\in\mathcal{N}$ and $n_1$ precedes it on the same
    strand in $\Sigma$, then $n_1\in\mathcal{N}$; 
    \item If $n_2$ is a reception node, there is exactly one
    $n_1\in\mathcal{N}$ such that $n_1\rightarrow n_2$; and
    \item The transitive closure $(\Rightarrow\cup\rightarrow)^+$ of
    the two arrow relations is acyclic.
  \end{enumerate}
  $\bnd$ is a bundle \emph{of protocol} $\Pi$ iff every strand with
  nodes in $\bnd$ is either an instance of a role of $\Pi$, or else an
  instance of one of the adversary roles in 
  Fig.~\ref{fig:adversary roles}.
\end{definition}
Any finite behavior should have these properties, since otherwise some
participant starts a role of the protocol in the middle, or receives a
message no one sent, or else the (looping) pattern of events is
causally impossible.  By acyclicity, every bundle determines a partial
ordering $\preceq_\bnd$ on its nodes, where $n_1\preceq_\bnd n_2$
means that some path of one or more arrows $\rightarrow,\Rightarrow$
leads from $n_1$ to $n_2$ in $\bnd$.

We incorporate state transition histories directly into the bundles.
To do this, we enrich the bundles with a new relation $\leadsto$ that
propagates the current state from one event to another.  We do this so
that our analysis method can work with a single object that has both
message dependencies and state dependencies within it.  We also
distinguish between \emph{state transitions} and \emph{state
  observations}.  Transitions need to be linearly ordered if they
pertain to a single device, but many state observations may occur
between a single pair of state transitions.  They are like \emph{read}
events in parallel computation:  There is no need for concurrency
control to sequentialize their access to the state, as long as they
are properly nested between the right transition events.

This is an advantage of the strand space approach, which focuses on
partially ordered execution models.  It is important for
enrich-by-need analysis, where the exponential number of interleavings
must be avoided.

Later in this section, we will introduce a model containing a number
of traditional state machines, where we correlate the synchronization
nodes with transitions in their state histories.  We make this model
more rigorous in
\ifarchive{Section~\ref{sec:executions:b:c:phi}}\fi, where we prove an exact
match between the state respecting behaviors we use here and the more
traditional model of state machine histories.


%

\subsection{Enriching Bundles with State}
\label{sec:state_axioms:state:enriched}

We now enrich the bundles to incorporate states, and to propagate them
from node to node, just as transmissions and receptions propagate
messages.



The diagrams in Section~\ref{sec:envelope} suggest a way to
incorporate state into bundles:  We enrich them so that each state
synchronization event is associated with messages representing states.
A transition event is associated with a pair, representing the
pre-state before the transition together with the post-state after it.
The pre-state must be obtained from an earlier synchronization event.
The post-state is produced by the transition, and may thus be passed
to later events.  We also now distinguish state observation events;
these are associated with a single state, which is like a pre-state
since it is received from an earlier event that produced it.  We also
identify initiation events, which initialize a devices state and serve
as the beginning of a state computation history.

\begin{description}
  \item [Initiation] nodes $\circ\!\leadsto\!s$ record the event of
  creating a new state.  We use $\stinit s$ to indicate an initiation
  of state to $s$.
  \item[Observation] nodes $s\!\leadsto\!\circ$ record the current
  state without changing it.  We use $\obsv s$ to indicate an
  observation of state $s$.
  \item[Transition] nodes $s_0\!\leadsto\!\circ\!\leadsto\!s_1$
  represent the moment at which the state changes from a specific
  pre-state to a specific post-state.  We use $\tran(s_0,s_1)$ to
  indicate a state transition with pre-state $s_0$ and post-state
  $s_1$.
\end{description}
In specifying protocols and their state manipulations, we can use the
style illustrated in Fig.~\ref{fig:TPM roles}.  There, an observation
such as the synchronization node in the \texttt{quote} role, acquires
a message on the incoming $\leadsto$ arrow.  In this case, it is a
variable $s$, which is itself a parameter to the role which
contributes to the subsequent transmitted message.  The
\texttt{decrypt} role also has an incoming $\leadsto$ arrow labeled
with $s$; in this case, the role can proceed to engage in this event
only if the value $s$ equals a previously available parameter acquired
in the previous reception node.  The \texttt{extend} role has a
transition node, in which any pre-state $s$ will be updated to a new
post-state by hashing in the parameter $n$.  

These pre- and post-state annotations, using parameters that appear
elsewhere in the roles, determine subrelations of the transition
relation associated with each instance of a role.  An instance of the
\texttt{extend} role with a particular value $n_0$ for the parameter
$n$ will engage only in state transformations that hash in that value
$n_0$.    

Observation events are not strictly necessary; we could model the
checking of a state value as a transition
$s\!\leadsto\!\circ\!\leadsto\!s$.  However, this would require
observation events be ordered in a specific sequence.  This violates
the principled choice that our execution model not include unnecessary
ordering.

In the Introduction, we defined a protocol to be a finite set of
strands called the \emph{roles} of the protocol.  An \emph {enriched
  protocol} $\Pi^+$ will be a protocol $\Pi$ enriched with a
classification of its state synchronization events into
$\xstinit,\xtran$, and $\xobsv$ nodes, with each of those annotated
with messages defining their pre- and post-states.  The \emph{regular
  strands} of $\Pi^+$ are all of the substitution instances of the
roles of $\Pi^+$, including the instances of the pre- and post-states
on the synchronization nodes.  

An enriched bundle uses $\leadsto$ arrows to track the propagation of
the state of each device involved in the behavior.  This is not a
sufficient model for reasoning about state, which requires also the
two axioms of Defn.~\ref{def:bundle:state:respecting}, but it provides
the objects from which we will winnow the state-respecting bundles.

\begin{definition}[Enriched bundles] 
  \label{def:bundle:enriched}
  $\bnd^+=(\mathcal{N},\rightarrow,\leadsto)$ is an \emph{enriched
    bundle} iff $(\mathcal{N},\rightarrow)$ is a bundle, and moreover:

  \begin{enumerate}
    \item $n_1\leadsto n_2$ implies that $n_1$ is an $\xstinit$ or
    \label{clause:leadsto:init}
    $\xtran$ event and $n_2$ is an $\xobsv$ or $\xtran$ event, and the
    post-state of $n_1$ equals the pre-state of $n_2$;
    \item For each $\xobsv$ or $\xtran$ event $n_2$, there exists
    \label{clause:leadsto:unique}
    a unique $n_1$ such that $n_1\leadsto n_2$;
    \item The transitive closure
    $(\Rightarrow \cup \to \cup \leadsto)^+$ of the three arrow
    relations is acyclic.  We refer to the partial order it determines
    as $\prec_{\bnd^+}$ or $\prec$ when ${\bnd^+}$ is clear.
  \end{enumerate}
\end{definition}
%

Enriched bundles are not a sufficient execution model, however,
because they do not capture what is essentially different about state
as compared to messages: the way that the next transition event
consumes a state value, such that it cannot be available again unless
a new transition creates it again.  We can see this by connecting our
current set-up to a state-machine model.  

\ifarchive{
  \input{comp_model}}
\fi

Each enriched protocol $\Pi^+$ determines a type of state machine.
Its states (included in the set of messages) are all pre-states and
post-states of the synchronization nodes of all instances of the roles
of $\Pi^+$.  A state machine has a set of initial states.  In the
state machine determined by $\Pi^+$, the initial states are the states
$\sigma(s)$ such that some role $\rho\in\Pi^+$ has an initiation event
$\stinit s$, and  $\sigma$ is a substitution determining an instance
of $\rho$.   

The state machine determined by $\Pi^+$ has the state transition
relation $\trans$ consisting of all pairs of states $(s_1,s_2)$ where
\begin{description}
  \item[$s_1\trans s_2$] iff there exists a state transition node of
  $\Pi^+$ with pre-state $t_1$ and post-state $t_2$ and a substitution
  $\sigma$, such that $s_1=\sigma(t_1)$ and $s_2=\sigma(t_2)$.  
\end{description}
A \emph{state history} or \emph{computation} is a finite or infinite
sequence of states $s_0,s_1,\ldots$ that starts with an initial state
$s_0$, and, for every $i$, if $s_{i+1}$ is defined then
$s_i\trans s_{i+1}$.

\ifarchive{}\fi

The enriched bundles are not a sufficient model for reasoning about
state, because there are enriched bundles that do not correspond to
any execution in this sense.  We will illustrate this in
Section~\ref{sec:state_based_analysis}.

\subsection{Our Axioms of State}
\label{sec:state_axioms:axioms}

The initiation and transition events are meant to describe the
sequence of states that a device passes through.  The notion of bundle
says nothing about the ``out-degree'' of an event.  A message
transmission event can satisfy more than one message reception.
However, a state event (initiation or transition) can satisfy \emph{at
  most one} state transition event.

Observations must occur in a constrained place in the sequence of
states.  They acquire an incoming $\leadsto$ arrow from a transition
or an initiation.  Any such observation occurs before a subsequent
change in the state.  

These two principles---that transitions do not fork, and observations
must precede a transition that consumes their state---motivate our
execution model.  They are illustrated in
Fig.~\ref{fig:sequential semantics}.
\begin{figure}[tb]
  \begin{center}\tabcolsep=2em
    \begin{tabular}{cc}
      \xymatrix@R=5ex@C=.7em{
      (1)  &\circ\ar@{~>}[dr]\ar@{~>}[dl]\\ \xtran&=&\xtran
      } & 
      \xymatrix@R=5ex@C=.7em{
      (2)  &\circ\ar@{~>}[dr]\ar@{~>}[dl]\\ \xobsv & \prec & \xtran
      }
    \end{tabular}
  \end{center}
  \caption{State-respecting semantics.  (1) State produced (either from a
    $\xtran$ or $\xstinit$ event) cannot be consumed by
    two distinct transitions.  (2) Observation
    occurs after the state observed is produced but before that state is
    consumed by a subsequent transition.}\label{fig:sequential semantics}
\end{figure}
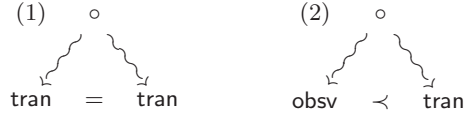

\begin{definition}[State-respecting bundle] 
  \label{def:bundle:state:respecting} 
  Let $\bnd^+=(\mathcal{N},\rightarrow,\leadsto)$ be an enriched
  bundle with precedence order $\prec$.  $\bnd^+$ is
  \emph{state-respecting} if and only if:
\begin{enumerate}
  \item\label{item:no state split} if $n \leadsto n_0$ and
  $n \leadsto n_1$, where $n_0$ and $n_1$ are $\xtran$ events, then
  $n_1 = n_0$;
  \item\label{item:obsv ordered} Let the relation $\prec^+$ be the
  smallest transitive relation including $\prec$
  such that
  whenever $n_0$ is an $\xobsv$ and $n_1$ is a $\xtran$, then
  \begin{equation}
    \label{eq:obs:tran:prec}
    n \leadsto n_0 \mbox{ and } n \leadsto n_1 \mbox{ implies } n_0
    \prec^+ n_1 . 
  \end{equation}
  Then $\prec^+$ is acyclic.
\end{enumerate}
\end{definition}
We call Clause~\ref{item:no state split} the \emph{No State Split
  Principle}.  Clause~\ref{item:obsv ordered} is the \emph{Observation
  Ordering Principle}.

These two axioms are adequate to provide a model of state.  In
particular, \ifarchive{}\fi
we \ifarchive{now }\fi 
prove that the executions in the sense we formalize there correspond
exactly to the state-respecting bundles of
Def.~\ref{def:bundle:state:respecting}.  
\ifarchive{\par
  \input{proofs}}\fi

\subsection{Enrich-by-need for stateful protocols}

In order to analyze stateful protocols with respect to
state-respecting bundles (Def.~\ref{def:bundle:state:respecting}), we
adapted the Cryptographic Protocol Shapes Analyzer (\cpsa) which
performs automated protocol analysis with respect to (traditional)
bundles (Def.~\ref{def:bundle}). {\cpsa} uses the enrich-by-need
method as described in the Introduction. That is, it progressively
extends an execution fragment $\skel$ into a set of execution
fragments $\{\skelB_i\}$. The extending occurs only as needed, namely,
when the execution fragment does not contain enough information to
fully describe a bundle. For message-only protocols, extending is
necessary exactly when a message received at node~$n$ cannot be
derived by the adversary using previously sent messages as inputs to a
web of adversary strands.

We adapted {\cpsa} in several ways to account for the properties of
state synchronization nodes in state-respecting bundles. First, we
added state synchronization nodes to the internal data structures of
the tool. We then augmented the tool to recognize that extending is
necessary when a state synchronization node~$n$ has pre-state~$s$, but
there is no node $n_0$ with post-state~$s$ such that $n_0\leadsto n$.
Finally, we implemented the corresponding rules for extending
execution fragments by adding state synchronization nodes that supply
the necessary state. In doing so, we experimented with two versions,
one works for enriched bundles that need not satisfy the two axioms
from Def.~\ref{def:bundle:state:respecting}, and one which enforces
these axioms. This former version allow us to perform analyses that
lead to bundles satisfying Def.~\ref{def:bundle:enriched} which do not
correspond to any executions of the state-machine model.  The latter
eliminates these ersatz results.

One advantage to the use of state-respecting bundles is that it
allowed us to integrate an analysis of the stateful part of the
protocol in a modular fashion. Our current release of
{\cpsa}~\cite{cpsa09} simply adds techniques for state-based reasoning
without altering the message passing analysis algorithms.  The analysis
of protocols that do not contain state synchronization nodes remains
unchanged.  We thus provide a clean separation of the two distinct
aspects of stateful protocols in an integrated whole.

The next section explores several examples that demonstrate the
results of these two versions and hopefully provide some intuition
about why the two axioms of state are necessary.


%% file: comp_model.tex
\label{sec:executions}

\subsection{Bundles with Explicit Computations}
\label{sec:executions:b:c:phi}

We introduce a formal model of executions, where protocol behavior
drives state machine executions, as briefly introduced in
Section~\ref{sec:state_axioms}.  Message transmissions and receptions
occur alongside the state transition histories of zero or more
stateful devices.  The message behavior here satisfies the usual
bundle properties for protocol behavior (see Def.~\ref{def:bundle}).
Some events do not send or receive messages, but synchronize with the
state of one or more devices.  Thus, a bundle together with a family
of state transition histories counts as a possible execution if the
steps of the state transition histories match with the state
synchronization events in the bundle.  This model is adapted from our
earlier work~\cite{Guttman12}.

When a protocol executes in coordination with devices that maintain
state, the execution must have the structure of a bundle, as far as
the message-passing behavior is concerned, and must \emph{also} meet
the constraints that the devices impose.  Each device must undergo a
possible state transition history, and each transition should be
caused by something, namely by some state synchronization event in the
bundle.  

%
\begin{definition}[Computations]
  Fix a set of states $\states$ with a distinguished subset
  $\initstates\subseteq\states$ of initial states and a transition
  relation $\trans\subseteq\states\times\states$.
  \begin{enumerate}
    \item A \emph{state transition history} or \emph{computation} is a
    finite sequence of states $\compute=\seq{s_0,s_1,\ldots,s_\ell}$
    starting with the initial state $s_0\in\initstates$ and, for every
    $j$, if $0\le j<\ell$, then $s_j\trans s_{j+1}$.
    \item Regard $\{D_i\}_{i\in I}$ as a family of instances of the
    state machine, indexed by the members of an index set $I$.  

    A $\{D_i\}_{i\in I}$\emph{-family} of computations is a family
    $\{\compute_i\}_{i\in I}$ of computations indexed by the set $I$.
  \end{enumerate}
\end{definition}
Thus, in our model all of the devices are instances of a single type
of state machine, with its state space and transition relation fixed.
This is not a real limitation, since given a family of different types
of state machines, we can construct a single encompassing machine
type.  Its states are the disjoint union (or tagged union) of the
states of the individual machines.  The transition relation is also a
union of the individual machines' transition relations.  The initial
state of any particular run of a machine then determines which
component machine it will simulate thereafter.

A \emph{correlation} is a function that delivers a synchronization
node in a bundle for each step in a computation in some family.
\begin{definition}[Correlations]
  \label{def:correlation}
  Let $\bnd$ be a bundle, with synchronization nodes
  $\syncnodes(\bnd)$, and let $\{\compute_i\}_{i\in I}$ be a
  $\{D_i\}_{i\in I}$-family of computations.  

  A \emph{position} $p=i,j$ for $\{\compute_i\}_{i\in I}$ is a pair
  such that $i\in I$ and $0\le j<\length{\compute_i}$.  Let $\Pos$ be
  the set of all positions for $\{\compute_i\}_{i\in I}$.

  A \emph{correlation} $\phi\colon \Pos\rightarrow\syncnodes(\bnd)$ is
  a function from positions in the computation family to
  synchronization nodes of the bundle and such that:
  \begin{enumerate}
    \item $\ran(\phi)=\syncnodes(\bnd)$, i.e.~$\phi$ is surjective
    onto the synchronization nodes; and
    \item $\phi$ is consistent with the bundle ordering
    $\prec_\bnd$: \label{clause:weak:order}
    i.e.~let $R(n,n')$ mean that there exist $i,j,k$ with $j<k$,
    $n=\phi(i,j)$, and $n'=\phi(i,k)$, and require:
    $$ (\prec_\bnd\cup\, R)^+ \mbox{ is acyclic} .  $$
  \end{enumerate}
  A correlation $\phi$ is \emph{injective} iff $\phi(i,j)=\phi(i',j')$
  implies $i=i'$ and $j=j'$.    
\end{definition}
In general, the same node $n$ may synchronize with positions in
several different computations ${\compute_i}$; an \emph{injective}
correlation does not exercise this possibility.  
%

Typically, one would like to correlate nodes and state transitions
more tightly, so that each synchronization node in a role causes a
specific type of transition.  In this context, a ``type'' of
transition simply means a subset of the transition relation.  The
subset can also depend on the parameter values for the node in
question.  The set $T$ of pairs
\begin{eqnarray*}
  \{n,\sigma & \colon & n\in\syncnodes(\Pi),\; \\ 
             && \mbox{ substitution } \sigma \mbox{ assigns values to the
                parameters of }\rho\} 
\end{eqnarray*}
indexes a family of subrelations $\trans^{n,\sigma}$ of a transition
relation $\trans$, i.e.~$\trans^{n,\sigma}\subseteq\trans$.
\begin{definition}[Execution] 
  \label{def:b:c:phi} 
  Let $\bnd$ be a $\Pi$-bundle; let $\{\compute_i\}_{i\in I}$ be a
  $\{D_i\}_{i\in I}$-family of computations; and let $\phi$ be a
  correlation between them.  For each $i\in I$, let
  $\trans^{n,\sigma}$ be a family of subrelations of $\trans$.

  \begin{description}
    \item[$(\bnd,\{\compute_i\}_{i\in I},\phi)$] is a
    $\Pi,\trans$-\emph{execution} acting on the devices
    $\{D_i\}_{i\in I}$, subject to the subrelations
    $\trans^{n,\sigma}$, iff for every $n'\in\syncnodes(\bnd)$,
    \begin{description}
      \item[if] $n'=\phi(i,j)=\sigma(n)$ is an instance of role node
      $n$ under substitution $\sigma$, 
      \item[then,] letting $\compute_i=\seq{s_0,s_1,\ldots,s_\ell}$,
      \\
      (i) when $j=0$, then $n$ is $\stinit s$ and $s_0=\sigma(s)$; and \\
      (ii) when $j>0$, then $s_{j-1}\trans^{n,\sigma}s_j$.
    \end{description}
  \end{description}
\end{definition}
This model, which we sometimes call the $\bnd,\compute,\phi$
model, is a general view of how the events in a protocol execution can
drive the transitions of a family of devices.  As discussed in our
previous~\cite{Guttman12}, it accounts both for events in which the
protocol execution receives information out of the state and also for
events in which the protocol execution deposits information into the
state.  There are two main changes here vis-a-vis~\cite{Guttman12}.
First, we allow many devices to have separate state histories.
Second, we omit the ``labels'' that were attached to synchronization
nodes there, instead using the subrelations $\trans^{n,\sigma}$ to
correlate specific protocol events with types of state transition.


%% file: proofs.tex
\subsection{Relating State-Enriched Bundles to Executions}
\label{sec:executions:state:enriched}

The extended protocols $\Pi^+$ and state-respecting bundles relate
easily to the $\bnd,\compute,\phi$ model.  

We now define a state machine in terms of the state synchronization
nodes of $\Pi^+$.  The set of states is a subset of the messages,
namely all those that can appear as pre-state or post-state of any
synchronization node.  Letting $\sigma$ range over substitutions and
$n$ over all synchronization nodes of $\Pi$, $\states=$
$$\begin{array}{cll}
  &  \{\sigma(s)\colon & n \mbox{ is }\stinit s\} \\
  \cup & \{\sigma(s_1), \sigma(s_2) \colon & n \mbox{ is }\tran
         (s_1,s_2)\} \\
  \cup & \{\sigma(s)\colon & n \mbox{ is }\obsv s\} . 
\end{array}
$$
$\initstates$ is the set given on the first line, the states
$\sigma(s)$ such that $n$ is $\stinit s$.  The transition relation
$\trans$ is determined from the non-initiation synchronization nodes
of $\Pi^+$.  If a node $n$ on a role of $\Pi^+$ is an
\begin{description}
  \item[Observation] node $n=t\!\leadsto\!\circ$, then
  $\trans^{n,\sigma}$ is the singleton
  $\{\seq{\sigma(t),\sigma(t)}\}$, in which the post-state is
  unchanged.
  \item[Transition] node $n=t_0\!\leadsto\!\circ\!\leadsto\!t_1$, then
  $\trans^{n,\sigma}$ is the singleton
  $\{\seq{\sigma(t_0),\sigma(t_1)}\}$.
\end{description}
Then the transition relation $\trans(\Pi^+)$ for the devices $D_i$ is
defined to be the union
$$ \bigcup_{n,\sigma}\trans^{n,\sigma} , 
$$
taking the union over all $\xobsv,\xtran$ nodes
$n\in\syncnodes(\Pi^+)$, and all substitutions $\sigma$.  We consider
executions relative to the family of subrelations $\trans^{n,\sigma}$.

If $\Pi^+$ is an extended protocol, let $\fgt(\Pi^+)$ result from it
by forgetting the pre- and post-state annotations on all
synchronization nodes.  We will refer to a node on some role
$\rho\in\fgt(\Pi^+)$ as an $\xstinit,\xobsv$, or $\xtran$ node of
$\fgt(\Pi^+)$ if it results from a node of the same kind in $\Pi^+$ by
forgetting.

If $\bnd^+=(\mathcal{N},\rightarrow,\leadsto)$ is an enriched bundle
for $\Pi^+$, then $\fgt(\bnd^+)$ is the $\Pi$-bundle
$(\mathcal{N}',\rightarrow')$ where $\mathcal{N}'$ results from
$\mathcal{N}$ by forgetting the pre- and post-state annotations on the
synchronizations, and $\rightarrow'$ relates two nodes in
$\mathcal{N}'$ iff $\rightarrow$ related their preimages in
$\mathcal{N}$.

%
%

\begin{lemma}
  \label{lemma:execution:to:sr:bundle}
  Given an extended protocol $\Pi^+$, let $\Pi=\fgt(\Pi^+)$ and
  $\trans=\trans(\Pi^+)$.  Let $\bnd,\{\compute_i\}_{i\in I},\phi$ be
  a $\Pi,\trans$-execution, with injective $\phi$.

  There exists a state-respecting bundle $\bnd^+$ of $\Pi^+$ such that
  $\fgt(\bnd^+)=\bnd$.
\end{lemma}
\begin{proof}
  Suppose that $\bnd,\{\compute_i\}_{i\in I},\phi$ is an execution.
  For each synchronization node in $\bnd$, we must decorate it with
  pre- and post-states (depending on its kind) from
  $\{\compute_i\}_{i\in I}$ and $\leadsto$ arrows, obtaining a
  state-respecting bundle $\bnd^+$.  It will then be immediate that
  $\fgt(\bnd^+)=\bnd$, since we will only add arrows and annotations
  that $\fgt$ discards.

  The correlation $\phi$ tells us how to decorate the nodes.  Since
  $\phi$ is surjective onto $\syncnodes(\bnd)$, every
  $n\in\syncnodes(\bnd)$ will be annotated.  Since $\phi$ is
  injective, there is no risk of conflict between two different
  computation steps.

  \medskip\noindent\textbf{Construction:}  We consider each
  computation $\compute_i$ and work recursively on steps $j$, where
  $0\le j<\length{\compute_i}$, within $\compute_i$.

  In the base case, when $j=0$, we know that the value $\compute_i(0)$
  was an initial state $s_0\in\initstates$; by Def.~\ref{def:b:c:phi},
  we know that $\phi(i,0)$ is an instance of an initiation node of
  $\Pi$.  Thus, we decorate $\phi(i,0)$ with
  $\stinit (\compute_i(0))$.

  Suppose, for the step case, that $j>0$.  We need now to decorate
  $\phi(i,j)$ with a pre-state and possibly post-state, and we need to
  provide it with an incoming $\leadsto$ arrow.

  For the $\leadsto$ arrow, if $\phi(i,j-1)$ is an initiation or
  transition node, we add an arrow $\phi(i,j-1)\leadsto\phi(i,j)$.  If
  $\phi(i,j-1)$ is an observation node, then it has an incoming arrow
  from some $n_1\leadsto\phi(i,j-1)$, and we add an arrow from the
  same $n_1\leadsto\phi(i,j)$.

  If $\phi(i,j)$ is an instance of an observation node of $\Pi$, by
  the decomposition of $\trans$ into subrelations, we know that
  $\compute_i(j-1)=\compute_i(j)$.  We decorate $\phi(i,j)$ with
  $\obsv (\compute_i(j))$.
  If instead $\phi(i,j)$ is not an instance of any observation nodes
  of $\Pi$, we decorate $\phi(i,j)$ with $\tran (\compute_i(j-1),
  \compute_i(j))$. 
 
  \medskip\noindent\textbf{Invariants:}  Our construction maintains
  the following invariants:
  \begin{enumerate}
    \item Whenever $n_1\leadsto n_2$, the post-state of $n_1$ is
    \label{clause:pre:post}
    well-defined, the pre-state of $n_2$ is well defined, and the two
    states are equal.
    \item Every $\leadsto$ arrow points from $\phi(i,j)$ to
    \label{clause:leadsto:forward}
    $\phi(i,k)$ where $j<k$ and the first arguments are equal.
    \item If $\phi(i,j)\leadsto\phi(i,k)$, then $j$ is the largest
    \label{clause:leadsto:max}
    $j'<k$ such that $\phi(i,j')$ is an $\xstinit$ or $\xtran$ node
    and the post-state of $\phi(i,j')$ equals the pre-state of
    $\phi(i,k)$.
    \item The set of indices $\{k\colon\phi(i,j) \leadsto \phi(i,k)\}$
    \label{clause:leadsto:interval}
    of $\leadsto$ successors of $\phi(i,j)$ forms an interval
    $[j+1,\ell]$; if $j<k<\ell$, then $\phi(i,k)$ is an $\xobsv$ node.
  \end{enumerate}
  Invariant~\ref{clause:leadsto:max} helps us to infer that
  invariant~\ref{clause:leadsto:interval} holds.  

  \medskip\noindent\textbf{$\bnd^+$ is state-respecting:}  By
  invariant~\ref{clause:pre:post}, Clause~\ref{clause:leadsto:init} of
  Def.~\ref{def:bundle:enriched} is satisfied.  In the construction,
  if $n=\phi(i,j)$ and $n$ is not an $\xstinit$ node, then $j\not=1$,
  so $n$ obtains a single incoming $\leadsto$ arrow.  So
  Clause~\ref{clause:leadsto:unique} is satisfied.  Moreover, by
  Clause~\ref{clause:weak:order} of Def.~\ref{def:correlation},
  $(\rightarrow\cup\Rightarrow\cup\leadsto)^+$ is acyclic.

  Thus, the resulting $\bnd^+$ is an enriched bundle.  We must now
  show that it satisfies the two axioms of state in
  Def.~\ref{def:bundle:state:respecting}.

  Suppose then that $n\leadsto n_1$ and $n\leadsto n_2$ forms a
  state-split, where $n_1,n_2$ are distinct $\xtran$-nodes.  By
  surjectiveness and invariant~\ref{clause:leadsto:forward},
  $n=\phi(i,j)$, $n_1=\phi(i,k)$, and $n_2=\phi(i,k')$, where
  $j<k,k'$.  By symmetry, we may assume $j<k<k'$.  But then by
  invariant~\ref{clause:leadsto:interval}, $\phi(i,k)$ is an $\xobsv$
  node contrary to assumption.  Thus the No State Split Principle
  (Def.~\ref{def:bundle:state:respecting}, Clause~\ref{item:no state
    split}) is satisfied.

  Turning to Clause~\ref{item:obsv ordered}, the Observation Ordering
  Principle, consider the set $S$ of pairs $n_0,n_1$ such that $n_0$
  is an $\xobsv$ node, $n_1$ is a $\xtran$ node, and for some $n$,
  $n\leadsto n_0$ and $n\leadsto n_1$.  Then for each such pair, by
  invariant~\ref{clause:leadsto:interval}, we have $n=\phi(i,j)$,
  $n_0=\phi(i,k)$, $n_1=\phi(i,k')$ where $j<k<k'$.  Therefore, we
  have $S\subseteq R$ for the $R$ in Clause~\ref{clause:weak:order}
  for the correlation $\phi$ (Def.~\ref{def:correlation}).  Thus,
  acyclicity follows.  \qed
\end{proof}
 
\begin{lemma}
  \label{lemma:sr:bundle:to:execution}
  Given an extended protocol $\Pi^+$, let $\Pi=\fgt(\Pi^+)$ and
  $\trans=\trans(\Pi^+)$.  Let $\bnd^+$ be a state-respecting bundle
  of $\Pi^+$.

  There exists a $\Pi,\trans$-execution
  $\fgt(\bnd^+),\{\compute_i\}_{i\in I},\phi$ with $\phi$ injective.
\end{lemma}

\begin{proof} \textbf{Determining $I$ and partitioning the nodes.}  By
  well-founded induction, for every synchronization node $n_1$, there
  exists an initiation node $n_0$ such that $n_0\leadsto^*n_1$.  By
  Def.~\ref{def:bundle:enriched}, Clause~\ref{clause:leadsto:init}, if
  $n_0$ is an initiation node, then for all $n$, $n\not\leadsto n_0$.
  By induction and the uniqueness in Def.~\ref{def:bundle:enriched},
  Clause~\ref{clause:leadsto:unique}, there is exactly one initiation
  node $n_0$ such that $n_0\leadsto^*n_1$.  Define the index set
  $I=\{n_0\in\syncnodes(\bnd^+)\colon n_0$ is an $\xstinit\}$.  Now,
  the $I$-indexed family of sets of nodes:
  $$ \mathcal{P} = \{ \, \{n_1\colon n_0\leadsto^*n_1 \} \, 
  \}_{n_0\in I} $$ 
  is a partition of the synchronization nodes indexed by $\xstinit$
  nodes.

  Consider now an $\xstinit$ node $n_0$ and the partition element
  $P_{n_0}$ of $\mathcal{P}$ where
  $$ P_{n_0} = \{ n_1\colon n_0\leadsto^*n_1 \} .
  $$
  We will show how to construct a computation $\compute_{n_0}$ and a
  piece of the correlation $\lambda j\qdot\phi(n_0,j)$ that will cover
  $P_{n_0}$.

  \medskip\noindent\textbf{Ordering the nodes.}  First, consider the
  $\xtran,\xstinit$ nodes in $P_{n_0}$.  We claim that they are
  \emph{linearly ordered} by $\leadsto^+$.  For otherwise, let
  $n_1,n_2$ be distinct incomparable nodes under $\leadsto^+$.  By the
  definition of $P_{n_0}$, $n_1\not=n_0\not=n_2$, since $n_0$ is
  related to every node in $P_{n_0}$.  Thus, $n_1,n_2$ are $\xtran$
  nodes.  By well-foundedness, we may assume that $n_1,n_2$ are each
  chosen to be a $\leadsto^+$-minimal pair of incomparable $\xtran$
  nodes in $P_{n_0}$.  Since $n_1$ is an $\xtran$ node, it has a
  predecessor $n_1'$.  By minimality, either $n_1'\leadsto^+n_2$ or
  else $n_2\leadsto^+n_1'$.  But the latter implies
  $n_2\leadsto^+n_1$, so in fact $n_1'\leadsto^+n_2$.  By minimality
  of $n_2$, $n_1'\leadsto n_2$.  So $n_1'\leadsto n_1$ and
  $n_1'\leadsto n_2$, contradicting the No State Split Principle.

  So $\xtran,\xstinit$ nodes in $P_{n_0}$ are linearly ordered by
  $\leadsto^+$.  

  Consider any pair of adjacent $\xtran,\xstinit$ nodes $n_1\leadsto
  n_2$.  Let $\mathcal{O}$ be the set of $\xobsv$ nodes $n_o$ such
  that $n_1\leadsto n_o$.  $\mathcal{O}\cup\{n_1,n_2\}$ are partially
  ordered by $\prec^+$ since $\bnd^+$ satisfies the Observation
  Ordering Principle.  Let $n_1,o_1,\ldots,o_k,n_2$ be any
  linearization of this set compatible with $\prec^+$.  

  Applying this throughout $P_{n_0}$, we obtain a sequence
  $\seq{n_0, \ldots,n_\ell}$ containing all the nodes of $P_{n_0}$,
  where the sequence ordering extends the $\prec^+$ ordering.

  \medskip\noindent\textbf{Defining $\compute_{n_0}$ and
    $\lambda j\qdot\phi(n_0,j)$.}  We now define $\compute_{n_0}$ and
  the $n_0$ slice of $\phi$ by stipulating:
  \begin{enumerate}
    \item$\compute_{n_0}(0)=$ the post-state of the initiation node
    $n_0$;
    \item$\compute_{n_0}(j+1)$ is the post-state of $n_j$ if it is a
    $\xtran$ node, and the pre-state of $n_j$ if it is an $\obsv$
    node;
    \item$\phi(n_0,j)=n_j$ for all $j$ where
    $0\le j<\length{\compute_{n_0}}$.
  \end{enumerate}
  Now by the definition of $\seq{n_0, \ldots,n_\ell}$, $\phi$
  satisfies the order constraint (Clause~\ref{clause:weak:order}), and
  the other clauses for correlations are immediate; moreover, $\phi$
  is injective because the $\lambda j\qdot\phi(n_0,j)$ are disjoint
  for different partition classes $P_{n_0}$.  The triple
  $\fgt(\bnd^+),\{\compute_i\}_{i\in I},\phi$ is an execution of
  $\Pi,\trans(\Pi^+)$, by the definition of $\trans(\Pi^+)$.  \qed
\end{proof}


%% file: state_based_analysis.tex
\section{Analysis of the Envelope Protocol}
\label{sec:state_based_analysis}

The two conditions of Def.~\ref{def:bundle:state:respecting} identify
the crucial aspects of state that distinguish state events from
message events. They axiomatize necessary properties of state that are
not otherwise captured by the properties of enriched bundles. In order
to give the reader some intuition for these properties, we present
several analyses of the Envelope Protocol in this section. We begin by
contrasting two analyses; one is based on enriched bundles that only
satisfy Definition~\ref{def:bundle:enriched}, while the other is based
on state-respecting bundles that also satisfy
Definition~\ref{def:bundle:state:respecting}

\paragraph{Enriched vs. state-respecting bundles.}
Recall that the Envelope Protocol was designed to satisfy Security
Goal~\ref{goal:alice}. That is, there should be no executions in which
(1) Alice completes a run with fresh, randomly chosen values for~$v$
and~$n$, (2)~$v$ is available unencrypted on the network, and (3) the
refusal certificate $Q$ is also available on the network. Whether we
use enriched bundles or state-respecting bundles as our model of
execution, the analysis begins the same way. The relevant fragment of
the point at which the two analyses diverges is depicted in
Fig.~\ref{fig:state split}. The reader may wish to refer to the figure
during the following description of the enrich-by-need process. The
first three steps describe how we infer the existence of the top row
of strands from right to left. The last two steps explain how we infer
the strands in the bottom row from left to right.

\begin{figure}
  \[\xymatrix@R=3ex{\texttt{alice}\ar@{:>}[d]&\texttt{extend}\ar@{:>}[d]
    &\texttt{extend}\ar@{:>}[dd]&\texttt{decrypt}\ar@{:>}[dd]\\
    \bullet\ar[r]^{\cdots n\cdots}\ar@{:>}[ddd]&\bullet\ar@2[d]\\
    &\circ\ar@{~>}[r]\ar@{.>}[dddr]^{\ast}&\circ\ar@{~>}[r]&\circ\ar@2[d]\\
    &\texttt{quote}\ar@{:>}[d]&\texttt{extend}\ar@{:>}[dd]&\bullet\ar[r]^{v}&\\
    \bullet\ar[r]&\bullet\ar@2[d]\\
    &\circ\ar@{=>}[d]&\circ\ar@{~>}[l]\\
    &\bullet\ar_{Q}[l]}\]
  \caption{A crucial moment in the {\cpsa} analysis of the Envelope
    Protocol, demonstrating the importance of our first axiom of
    state.}\label{fig:state split}
\end{figure}
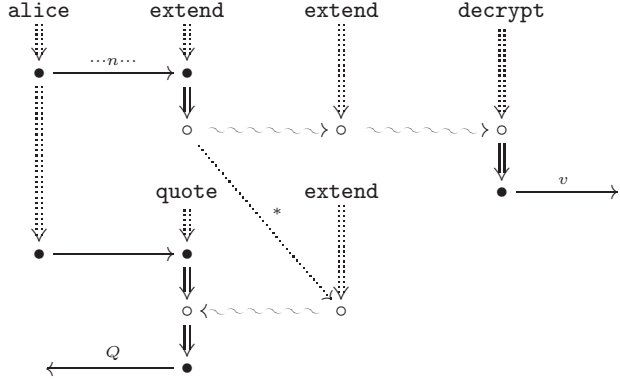

\begin{enumerate}
\item The presence of~$v$ in unencrypted form implies the existence of
  a \texttt{decrypt} strand to reveal it.
  \item The \texttt{decrypt} strand requires the current state to be
  $\#(\cn{obt},\#(n,s))$, so our new principle of state explanation
  implies the existence of an extend strand with input value
  $\cn{obt}$.
\item This newly inferred \texttt{extend} strand, in turn must have
  its current state $\#(n,s)$ explained which is done by another
  \texttt{extend} strand that receives the value $n$ from Alice. 
\item The presence of the quoted refusal token~$Q$ implies the
  existence of a \texttt{quote} strand to produce it.
\item The \texttt{quote} strand requires the state to be
  $\#(\cn{ref},\#(n,s))$, which allows us to infer the third
  \texttt{extend} strand.
\end{enumerate}

At this point in the analysis, the underlying semantics of bundles
begins to matter. Our analysis still must explain how the state became
$\#(n,s)$ for this last \texttt{extend} strand. If we use enriched
bundles that do not satisfy
Definition~\ref{def:bundle:state:respecting}, then we may re-use the
\texttt{extend} strand inferred in Step~3 as an explanation. This
would cause us to add a $\leadsto$ arrow between these two state
events (along the dotted arrow $\ast$ of Fig.~\ref{fig:state split})
forcing us to ``split'' the state coming out of the earlist extend
strand. Further steps allow us to discover an enriched bundle
compatible with our starting point, contrary to Security
Goal~\ref{goal:alice}.  Importantly, however, all enriched bundles
that extend the fragment with the split state are
non-state-respecting.

If, on the other hand, we only allow state-respecting bundles,
Condition~1 of Definition~\ref{def:bundle:state:respecting} does not allow
us to re-use the \texttt{extend} strand inferred in Step~3 to explain
the state found on the strand of Step~5. Instead, we are forced to
infer yet another \texttt{extend} strand that receives Alice's
nonce~$n$. However, since Alice uses an encrypted session that
provides replay protection, the adversary has no way to return the TPM
state to $\#(n,s)$. Thus, although there are enriched bundles that
violate Security Goal~\ref{goal:alice}, there are no state-respecting
bundles that do so. 

\paragraph{A flawed version.}  
We also performed an analysis of the Envelope Protocol, removing the
assumption that Alice's nonce $n$ is fresh, to demonstrate
our state-respecting variant's ability to automatically detect
attacks.  The analysis proceeds similarly; as in the previous analysis
we decline to add a $\leadsto$ arrow along $\ast$ thanks to our
stateful semantics.  However, the alternative possibility that a fresh
\texttt{extend} strand provides the necessary state proves to work
out.  Because $n$ is not freshly chosen, the parent can engage in a
distinct \texttt{extend} session with the same $n$.

Note that our analysis does not specify that $s = \cn{s}_0$, where
$s$ is the state of the PCR when first extended.  For the case where
$s = \cn{s}_0$, the attack is to reboot the TPM after obtaining one
value (either the refuse token or Alice's secret), re-extend the boot
state with $n$, and then obtain the other.  More generally, as long as $s$
is a state that the parent can induce, a similar
attack is possible.  


%% file: axiom2.tex
\subsection{The Importance of Observer Ordering}
\label{sec:axiom2}

The Envelope Protocol example demonstrates the crucial importance of
capturing our first axiom of state correctly.  The second axiom,
involving the relative order of observations and state transition, is
no less crucial to correct understanding of stateful protocols.

Another example protocol, motivated by a well-known issue with PKCS \#11 
(see, e.g. \cite{DKR08}), illustrates the principle more clearly.
Suppose a hardware device is capable of producing keys that are meant
to be managed by the device and not learnable externally.  If the
device has limited memory, it may be necessary to export such a key in
an encrypted form so the device can utilize external storage.

Thus, device keys can be used for two distinct purposes: for
encryption / decryption of values on request, or for encrypting
internal keys for external storage.  It is important that the purpose
of a given key be carefully tracked, so that the device is not induced
to decrypt one of its own encrypted keys.

Suppose that for each key, the device maintains a piece of state,
namely, one of three settings:

\begin{itemize}
\item A $\mathsf{wrap}$ key is used only to encrypt internal keys.
\item A $\mathsf{decrypt}$ key may be used to encrypt or decrypt.
\item An $\mathsf{initial}$ key has not yet been assigned to either use.
\end{itemize}

If a key in the $\mathsf{wrap}$ state can later be put in the
$\mathsf{decrypt}$ state, a relatively obvious attack becomes
possible: while in the wrap state, the device encrypts some internal
key, and later, when the key is in the decrypt state, the device
decrypts the encrypted internal key.

However, if keys can never exit the $\mathsf{wrap}$ state once they
enter it, this attack should not be possible.  If we were to represent
this protocol within {\cpsa}, we would include the following roles:

\begin{itemize}
\item A \texttt{create key} role that generates a fresh key and
  initializes its state to $\mathsf{initial}$
\item A \texttt{set wrap} role that transitions a key from
  $\mathsf{initial}$ or $\mathsf{decrypt}$ to $\mathsf{wrap}$.
\item A \texttt{set decrypt} role that transitions a key from
  $\mathsf{initial}$ to $\mathsf{decrypt}$.
\item A \texttt{wrap} role in which a user specifies two keys (by
  reference), and the device checks (with an observer) that the first
  is in the $\mathsf{wrap}$ state and if so, then encrypts the second
  key with the first and transmits the result.
\item A \texttt{decrypt} role in which a user specifies a key (by
  reference) and a ciphertext encrypted under that key, and the device
  checks (with an observer) that the key is in the $\mathsf{decrypt}$
  state and if so, then decrypts the ciphertext and transmits the
  resulting plaintext.
\end{itemize}

\begin{figure}
 \[\xymatrix@R=3ex{\texttt{init}\ar@{:>}[d]&\texttt{set decrypt}\ar@{:>}[d]&\texttt{set wrap}\ar@{:>}[d]&&\texttt{wrap}\ar@{:>}[d]\\
             \circ\ar@{~>}[r]&\circ\ar@{~>}[r]\ar@{~>}[dddrr]&\circ
\ar@{~>}[rr]&&\circ\ar@{=>}[d]\\
             &&&\texttt{decrypt}\ar@{:>}[d]&\bullet\ar[dl]^{\enc{k_2}{k_1}}\\
             &&&\bullet\ar@{=>}[d]&\\
             &&&\circ\ar@{=>}[d]\ar@{.>}[uuul]&\\
             &&&\bullet\ar[r]^{k_2}&\\
}\]
  \caption{Observer ordering example}\label{fig:wrap-decrypt}
\end{figure}
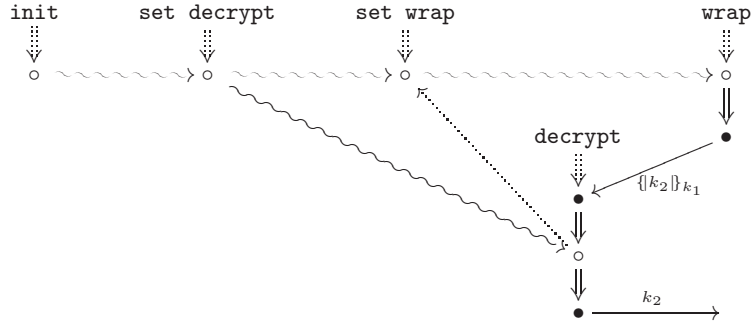

Note that the attack should not be possible.  However, the bundle
described in Fig.~\ref{fig:wrap-decrypt} is a valid bundle, and fails
to be state-respecting only because of our axiom about observers.  Our
second axiom induces an ordering so that the observer in the
\texttt{decrypt} strand occurs before the following transition event
in the \texttt{set wrap} strand.  The induced ordering is shown in the
figure with a single dotted arrow; note the cycle among state events
present with that ordering that is not present without it.


%% file: related.tex
\section{Related Work}
\label{sec:related}

The problem of reasoning about protocols and state has been an
increasing focus over the past several years.  Protocols using TPMs
and other hardware security modules (HSMs) have provided one of the
main motivations for this line of work.

A line of work was motivated by HSMs used in the banking
industry~\cite{herzog2006applying,youn2007robbing}.  This work
identified the effects of persistent storage as complicating the
security analysis of the devices.  There was also a strong focus on
the case of PKCS \#11 style devices for key
management~\cite{cortier2007automatic,cortier2009generic,froschle2011reasoning}.
These papers, while very informative, exploited specific
characteristics of the HSM problem; in particular, the most important
mutable state concerns the \emph{attributes} that determine the usage
permitted for keys.  These attributes should usually be handled in a
monotonic way, so that once an attribute has been set, it will not be
removed.  This justifies using abstractions that are more typical of
standard protocol analysis.

In the TPM-oriented line of work, an early example using an
automata-based model was by G{\"u}rgens et
al.~\cite{gurgens2007security}.  It identified some protocol failures
due to the weak binding between a TPM-resident key and an individual
person.  Datta et al.'s ``A Logic of Secure
Systems''~\cite{datta2009logic} presents a dynamic logic in the style
of PCL~\cite{DattaEtAl05} that can be used to reason about programs
that both manipulate memory and also transmit and receive
cryptographically constructed messages.  Because it has a very
detailed model of execution, it appears to require a level of effort
similar to (multithreaded) program verification, unlike the less
demanding forms of protocol analysis.

M{\"o}dersheim's set-membership
abstraction~\cite{modersheim2010abstraction} works by identifying all
data values (e.g.~keys) that have the same properties; a change in
properties for a given key $K$ is represented by translating all facts
true for $K$'s old abstraction into new facts true of $K$'s new
abstraction.  The reasoning is still based on monotonic methods
(namely Horn clauses).  %
Thus, it seems not to be a strategy for reasoning about TPM usage, for
instance in the Envelope Protocol.

Guttman~\cite{Guttman12} developed a theory for protocols (within
strand spaces) as constrained by state transitions, and applied that
theory to a fair exchange protocol.  It introduced the key notion of
\emph{compatibility} between a protocol execution (``bundle'') and a
state history.  This led to work by Ramsdell et
al. ~\cite{RamsdellDGR14} that used {\cpsa} to draw conclusions in the
states-as-messages model.  Additional consequences could then be
proved using the theorem prover PVS~\cite{cade92-pvs}, working within
a theory of both messages and state organized around compatibility.

A group of papers by Ryan with Delaune, Kremer, and
Steel~\cite{delaune2011formal,DelauneEtAl2011}, and with Arapinis and
Ritter~\cite{arapinis2011statverif} aim broadly to adapt ProVerif for
protocols that interact with long-term state.
ProVerif~\cite{Blanchet01,AbadiBlanchet05} is a Horn-clause based
protocol analyzer with a monotonic method: in its normal mode of
usage, it tracks the messages that the adversary can obtain, and
assumes that these will always remain available.  Ryan et al.~address
the inherent non-monotonicity of adversary's capabilities by using a
two-place predicate $\mathrm{att}(u,m)$ meaning that the adversary may
possess $m$ at some time when the long-term state is $u$.
In~\cite{arapinis2011statverif}, the authors provide a compiler from a
process algebra with state-manipulating operators to sets of Horn
clauses using this primitive.  In~\cite{DelauneEtAl2011}, the authors
analyze protocols with specific syntactic properties that help ensure
termination of the analysis.  In particular, they bound the state
values that may be stored in the TPMs.  In this way, the authors
verify two protocols using the TPM, including the Envelope Protocol.

Meier, Schmidt, Cremers, and Basin's tamarin prover~\cite{MeierSCB13}
uses multiset rewriting (MSR) as a semantics in which to prove
properties of protocols.  Since MSR suffices to represent state, it
provides a way to prove results about protocols with state.
K{\"u}nnemann studied state-based protocol
analysis~\cite{KremerKuennemann14} in a process algebra akin to
StatVerif, which he translated into the input language of tamarin to
use it as a proof method.  Curiously, the main constructs for mutable
state and concurrency control (locking) are axiomatized as properties
of traces rather than encoded within MSR
(see~\cite[Fig.~10]{KremerKuennemann14}).

\paragraph{Our work.}  One distinguishing feature of this work is
our extremely simple modification to the plain message
passing semantics to obtain a state-respecting model.  These are the
two Axioms~\ref{item:no state split}--\ref{item:obsv ordered} in
Def.~\ref{def:bundle:state:respecting}.  We think it is an attractive
characteristic of the strand space framework that state reflects such
a clean foundational idea.  Moreover, this foundational idea motivated
a simple set of alterations to the enrich-by-need tool {\cpsa}.  


%% file: goals.tex
\section{Protocol Security Goals}
\label{sec:goals}

The enrich-by-need analysis performed in our enhanced version of
\cpsa~is fully compatible with the language of goals found in previous
work such as~\cite{RoweEtAl2015}.  The goal language is based on two
classes of predicates: role-related predicates that relate an event or
  parameter value to its use within a specific protocol role, and 
  predicates that are protocol-independent
  and describe important properties of bundles.  The latter includes the
  ordering of events as well as assumptions about freshly chosen
  values and uncompromised keys.  
Both classes of predicates apply within state-respecting bundles in a
natural way.  The role-related predicates are sensitive only to the
position of an event in the sequence of events of a role, and to the
choice of parameter values in that instance of the role.  Indeed, nodes
that represent state transitions or observations are handled in
exactly the same way, since they have positions in the role and
parameter values in just the same way as the message transmission and
reception events.  


Thus, the state-respecting version of {\cpsa} can verify formulas
expressing security goals in exactly the same way as the previous
version, and with the same semantic definitions.  

\paragraph{Conclusion.}  In this paper, we have argued
that {\cpsa}---and possibly other formalized protocol analysis
methods---can provide reliable analysis when protocols are
standardized, even when those protocols are manipulating devices with
long-term state.  A core idea of the formalization are the two axioms
of Def.~\ref{def:bundle:state:respecting}, which encapsulate the
difference between a message-based semantics and the state-respecting
semantics.
